\documentclass{article}

\usepackage{amsmath}
\usepackage{amsthm}
\usepackage{amssymb}
\usepackage{amsfonts}
\usepackage{graphicx}

\usepackage{fullpage}

\usepackage{algorithm}
\usepackage{algpseudocode}
\usepackage{xspace}
\usepackage{multirow}
\usepackage{hyperref}

\usepackage{todonotes}

\newtheorem{thm}{Theorem}
\newtheorem{lem}{Lemma}
\newtheorem{prop}{Proposition}
\newtheorem{coro}{Corollary}
\theoremstyle{definition}
\newtheorem{defn}{Definition}

\theoremstyle{remark}

\newcommand{\set}[1]{\left\{#1\right\}}
\newcommand{\card}[1]{\left|#1\right|}
\newcommand{\bigoh}[1]{\mathcal{O}\left(#1\right)}
\newcommand{\epsagreeformula}[3]{\frac{\card{N_{#1}(#2) \Delta N_{#1}(#3)}}{\max\set{\card{N_{#1}(#2)},\card{N_{#1}(#3)}}}}
\newcommand{\nonagreement}[3]{\textsc{NonAgreement}_{#1}\left(#2,#3\right)}

\newcommand{\R}{\mathbb{R}}

\newcommand{\sign}{\textsc{Sign}\xspace}

\newcommand{\agreecnt}{\textsc{AgreeCnt}\xspace}
\newcommand{\epsagree}{$\varepsilon$-agreement\xspace}

\newcommand{\unprocessed}{\textsc{Unprocessed}\xspace}
\newcommand{\copycand}{\textsc{CopyCandidate}\xspace}
\newcommand{\mergecand}{\textsc{MergeCandidate}\xspace}
\newcommand{\splitcand}{\textsc{SplitCandidate}\xspace}

\DeclareMathOperator{\cost}{cost}
\DeclareMathOperator{\op}{\textsc{Op}}

\title{Online Correlation Clustering for Dynamic Complete Signed Graphs}
\author{Ali Shakiba}
\date{\footnotesize{Department of Computer Science, Vali-e-Asr University of Rafsanjan, Rafsanjan, Iran. \\ \footnotesize{\texttt{ali.shakiba@vru.ac.ir;a.shakiba.iran@gmail.com}}}}

\begin{document}
	\maketitle
	
	\begin{abstract}
		In the correlation clustering problem for complete signed graphs, the input is a complete signed graph with edges weighted as $+1$ (denote recommendation to put this pair in the same cluster) or $-1$ (recommending to put this pair of vertices in separate clusters) and the target is to cluster the set of vertices such that the number of disagreements with these recommendations is minimized.

In this paper, we consider the problem of correlation clustering for dynamic complete signed graphs where (1) a vertex can be added or deleted, and (2) the sign of an edge can be flipped. In the proposed online scheme, the offline approximation algorithm in \cite{clmnpt21-correlation-clustering-in-constant-many-parallel-rounds} for correlation clustering is used. Up to the author's knowledge, this is the first online algorithm for dynamic graphs which allows a full set of graph editing operations. 

The proposed approach is rigorously analyzed and compared with a baseline method, which runs the original offline algorithm on each time step. Our results show that the dynamic operations have local effects on the neighboring vertices and we employ this locality to reduce the dependency of the running time in the Baseline to the summation of the degree of all vertices in $G_t$, the graph after applying the graph edit operation at time step $t$, to the summation of the degree of the changing vertices (e.g. two endpoints of an edge) and the number of clusters in the previous time step. Moreover, the required working memory is reduced to the square of the summation of the degree of the modified edge endpoints rather than the total number of vertices in the graph. 

\noindent\textbf{Keywords:} Correlation Clustering, Online Algorithms, Signed Graph

	\end{abstract}
	\section{Introduction}
		
The goal of clustering, as a long-standing problem in machine learning, is to divide a set of elements such that in each partition or cluster, the elements are similar to each other and dissimilar to the elements in other clusters. The problem of correlation clustering is proposed by Bansal, Blum, and Chawla in \cite{bbc04-correlation-clustering} and found many applications in disambiguation tasks \cite{k08}, community detection \cite{c12}, cluster ensembles \cite{b13}, due to its natural simple formulation. In the original correlation clustering formulation, the input is a weighted graph with positive and negative edge weights. The negative edge weights are used to represent dissimilarity between vertices whereas the positive edge weights are used to represent the similarity between two vertices. Then, correlation clustering asks for clustering this graph such that the sum of the absolute values of the weights of the negative edges inside any cluster plus the sum of the weight of the positively weighted edges between any clusters is minimized. The complement to this minimization, one may ask for maximum agreement, which we refer to as the maximization instance of the correlation clustering. Unfortunately, this problem in both maximization and minimization variants is known to be $\mathbb{NP}$-hard. A simple variant of this problem is well-studied in the literature: the weights of the edges are either $+1$ or $-1$, i.e. the underlying graph is signed. The maximization instance of this variant is shown to have a polynomial-time approximation scheme in \cite{bbc04-correlation-clustering}. Moreover, in \cite{c15}, a $2.06$-approximation algorithm is given for the minimization version of this variant. In \cite{gg05}, an approximation algorithm with polynomial running time is given when the number of clusters is bounded to be no more than $k$ clusters. There is also an $\bigoh{\log n}$-approximation for the minimization instance of the problem where the weights are arbitrary and $n$ is the number of vertices, due to \cite{d06}. Also, there is a constant factor approximation scheme for solving the minimization instance of this problem in \cite{clmnpt21-correlation-clustering-in-constant-many-parallel-rounds} when the edge weights are restricted to $\set{-1,+1}$ and the underlying graph is complete. 

One main restriction to most of the previously stated algorithms is their reliance on Integer Programming techniques. This makes these techniques less attractive for dynamic graphs. In \cite{clmnpt21-correlation-clustering-in-constant-many-parallel-rounds}, a $\bigoh{1}$-factor approximation algorithm for solving the correlation minimization variant in signed graphs in a semi-streaming setting is proposed which has $\bigoh{1}$ passes over the dataset. Moreover, the same algorithm is then applied in \cite{clmp22-online-consistent-correlation-clustering} in an online setting to dynamic signed graphs for complete graphs. In both the aforementioned methods, adding a vertex that reveals all its incident edges to the existing vertices is the only allowed method and after that, no edge might change its sign or a vertex might be deleted. In this paper, we propose the first algorithm, up to the author's knowledge, for the correlation clustering problem on signed complete graphs where a full set of dynamic operations are allowed. 

\subsection{Our Contribution}
	Our main contribution is analyzing the behavior of the \textsc{CorrelationClustering} algorithm of \cite{clmnpt21-correlation-clustering-in-constant-many-parallel-rounds} in a full dynamic graph setting where one can add/remove vertices and flip edge signs. This is not the first extension of the \textsc{CorrelationClustering} algorithm to dynamics graphs, as up to the author's knowledge, it is extended to accomplish consistent clustering for dynamic graphs in \cite{clmp22-online-consistent-correlation-clustering}. However, the extension in \cite{clmp22-online-consistent-correlation-clustering} only allows adding vertices that reveal all of its edges to existing vertices, whereas in our setting, any change is possible. The main idea in our algorithm is tracking the changes to the underlying graph, which as we will show, is local. Our most important result is the following theorem.
	\begin{thm}
		\label{thm:result}
		The proposed \textsc{OnlineCorrelationClustering} algorithm (Algorithm \ref{alg:online:algorithm}) maintains the same clustering as the clustering by the offline \textsc{CorrelationClustering} algorithm (Algorithm \ref{alg:correlation:clustering} and described as \textsc{Baseline} in Section \ref{sec:proposed}) on each time step $t=1,2,\ldots$ with the following time and space complexities:
		\begin{itemize}
			\item Adding and deleting vertices with all negative signed incident edges can be accomplished in constant time with constant extra memory.
			\item Flipping the sign of an edge $e=\set{u,v}$ requires $$\bigoh{\card{V_t}+\card{E^+_t}+\left(\deg_{G^+_t}(u) + \deg_{G^+_t}(v)\right)^2+\card{\mathcal{C}_{t-1}}},$$ time with $\bigoh{\left(\deg_{G^+_t}(u) + \deg_{G^+_t}(v)\right)^2}$ extra working memory. 
		\end{itemize}
		Note that computing the initial offline clustering for the initial graph $G_0$ requires $\bigoh{\card{V_0} + \card{E^+_0} + D_0}$ time where $D_0=\sum_{e=\set{u,v}\in E^+_0}\left(\deg_{G^+_0}(u) + \deg_{G^+_0}(v)\right)$ time and $\bigoh{\card{V_0}}$ offline memory. 
	\end{thm}
	Up to the author's knowledge, this is the first online correlation clustering algorithm for signed graphs where edge modification in terms of its sign and vertex addition and deletion are allowed.
	In comparison with the \textsc{Baseline} algorithm, which requires $\bigoh{\card{V_t} + \card{E^+_t} + D_t}$ time where $$D_t=\sum_{e=\set{u,v}\in E^+_t}\left(\deg_{G^+_t}(u) + \deg_{G^+_t}(v)\right),$$ and $\bigoh{\card{V_t}}$ memory on each time step, our proposed approach has the following advantageous:
	\begin{itemize}
		\item Using the locality of changes made during the sequence of operations, we reduced the dependency of the running time in \textsc{Baseline} to the summation of the degree of all vertices in $G_t$ to the summation of the degree of the changing vertices (e.g. two endpoints of an edge) and the number of clusters in the previous time step. Note that it is expected that the number of clusters at any point in time is much smaller than the number of vertices of the graph. 
		\item Again, using the locality effect, we were able to reduce the required working memory to $$\bigoh{\left(\deg_{G^+_t}(u) + \deg_{G^+_t}(v)\right)^2},$$ compared with $\bigoh{\card{V_t}}$. Note that for real-world data, we expect very low degrees compared to the number of vertices, although $\deg_{G^+_t}(v) = \bigoh{\card{V_t}}$. 
	\end{itemize}

	\section{Preliminaries}\label{sec:pre}
		
In the min-disagree variant of the correlation clustering problem, we are given a signed undirected graph $G=(V,E^+,E^-)$ and the goal is finding a partitioning of the vertices $V$ into $\mathcal{C}=\set{C_1,\ldots,C_\ell}$ which minimizes 
\begin{equation}
	\label{eq:min:disagree:corr:clustering}
	\cost(\mathcal{C}) = \sum_{\substack{\set{u,v}\in E^+\\u\in C_i,v\in C_j,i\neq j}} 1 + \sum_{\substack{\set{u,v}\in E^-\\u,v\in C_i}} 1,
\end{equation}
where $E^+$ denotes the set of all edges in $G$ whose sign is $+$. We use $G^+=G[E^+]$ to denote the induced subgraph of $G$ with the same set of vertices and $E^+$ as the set of edges, i.e. $G^+=(V,E^+)$.

\begin{defn}[$\varepsilon$-agreement \cite{clmnpt21-correlation-clustering-in-constant-many-parallel-rounds}]
	\label{def:eps:agree}
	Let $G^+=G[E^+]$, $\varepsilon \in \R^+$, and $u,v\in V$. Then, vertices $u$ and $v$ are in $\varepsilon$-agreement if $u$ and $v$ are adjacent in $G^+$ and
	\begin{equation}
		\label{eq:eps:agreement}
		\nonagreement{G}{u}{v} = \epsagreeformula{G}{u}{v} < \varepsilon.
	\end{equation}
\end{defn}

\begin{defn}[$\varepsilon$-lightness \cite{clmnpt21-correlation-clustering-in-constant-many-parallel-rounds}]
	Let $G^+=G[E^+]$, $\varepsilon \in \R^+$, and $u\in V$. Then, the vertex $u$ is said to be $\varepsilon$-light, or simply light, in $G^+$ if 
	\begin{equation}
		\label{eq:lightness}
		\frac{\agreecnt_{G^+}(u)}{\card{N_{G^+}(u)}} < \varepsilon,
	\end{equation}
	where $\agreecnt_{G^+}(u)=\card{\set{w\in V|u \text{ and } v \text{ are in } \varepsilon\text{-agreement}}}$. Otherwise, it is called $\varepsilon$-heavy, or simply heavy. 
\end{defn}

The basic ingredient of our approach is the \textsc{CorrelationClustering} algorithm introduced in \cite{clmnpt21-correlation-clustering-in-constant-many-parallel-rounds}, which is restated in Algorithm \ref{alg:correlation:clustering}. This algorithm is a constant-factor approximation scheme for solving the min-disagree correlation clustering problem for a complete undirected signed graph. Let $\mathcal{C}_{\textsc{Opt}}$ be an optimal solution with cost $\textsc{Opt}_G$ (Eq. \eqref{eq:min:disagree:corr:clustering}) and $\mathcal{C}$ be the output of the Algorithm \ref{alg:correlation:clustering} with $G$ and $\varepsilon$ as its input. Then, $\cost(\mathcal{C}) \leq D + \textsc{Opt}_G$, where $D \leq 1+\frac{4}{\varepsilon}+\frac{1}{\varepsilon^2}$ (Theorem 3.9 in \cite{clmnpt21-correlation-clustering-in-constant-many-parallel-rounds}). In other words, \textsc{CorrelationClustering($G$)} is a $2+\frac{4}{\varepsilon}+\frac{1}{\varepsilon^2}$-approximation algorithm, or simple a constant factor approximation scheme. Note that in our presentation here, we substituted both $\beta$ and $\lambda$ with $\varepsilon$ and considered $i=1$ in comparison with \cite{clmnpt21-correlation-clustering-in-constant-many-parallel-rounds}. Moreover, \textsc{CorrelationClustering($G$)} is used in \cite{clmp22-online-consistent-correlation-clustering}, renamed to \textsc{Agreement} algorithm, for solving the min-disagree correlation clustering problem in an online setting. Their online approach only allows adding a vertex to the graph, which reveals all the edges to the previously arrived nodes upon arrival. Moreover, in \cite{clmp22-online-consistent-correlation-clustering}, they provide the \textsc{Online Agreement} algorithm which uses the \textsc{Correlation-Clustering} after each vertex arrival as an offline re-clustering phase. Their online approach is shown to be a constant-factor approximation algorithm, too (Theorem E.2 in \cite{clmp22-online-consistent-correlation-clustering}).
\begin{algorithm}
	\caption{\textsc{CorrelationClustering($G$)} \cite{clmnpt21-correlation-clustering-in-constant-many-parallel-rounds}}
	\label{alg:correlation:clustering}
	\begin{algorithmic}[1]
		\Procedure{CorrelationClustering}{$G,\varepsilon$}
			\State Let $G^+=G[E^+]$ where $E^+$ is the set of edges whose sign is $+$
			\State Discard all edges whose endpoints are not in \epsagree
			\State Discard all edges between two $\varepsilon$-light vertices
			\State Let $\widetilde{G^+}$ be the sparsified graph $G^+$ after performing previous two operations
			\State Let $\mathcal{C}$ be the collection of connected components in $\widetilde{G^+}$
			\State \Return $\mathcal{C}$ as the output clustering
		\EndProcedure
	\end{algorithmic}
\end{algorithm}
An efficient implementation of the algorithms in this paper is applicable using the adjacency list-representation of the graphs. Therefore, the neighborhood of each vertex is accessible in constant-time.

	\section{Proposed Method}\label{sec:proposed}
		
\paragraph{The Online Setting.} Consider an initial undirected complete signed graph $G_0=(V_0,E_0)$ and an infinite sequence of operations, one per time stamp $t=1,2,\ldots$ denoted as $\op_1,\op_2,\ldots$. We allow three kinds of operations: (1) $\textsc{FlipSign}(e)$ which flips the sign of an edge $e$, (2) $\textsc{AddVertex}(v)$ to add a new vertex $v$ to the graph whose all of its incident edges are negatively signed, and (3) $\textsc{DeleteVertex}(v)$ which deletes an existing vertex $v$ whose all of its incident edges are negatively signed. At first glance, the restriction on the vertex addition/deletion may seem restrictive, however, if one wants to delete a vertex with positively signed incident edges, it can use consecutive calls to $\textsc{FlipSign}(e)$ for all of its positively signed incident edges and finally, remove it. The same applies for the vertex addition.

\paragraph{Notation in Algorithms.} For every vertex $v$ in the graph $G$, we need three quantities: (1) $v.\textsc{IsLight}_{G^+}$ which is a boolean value denoting whether the vertex $v$ is light or not, (2) $v.\agreecnt_{G^+}$ as the number of vertices which are in $\varepsilon$-agreement with node $v$, and (3) $v.\textsc{Neigh}_{G^+}$ which is equal to $N_{G}(v)$. For each edge $e=\set{u,v}$, we also store its sign as $e.\sign$ (either $+$ or $-$), and $e.\textsc{Agree}$ which is a boolean equal to \textsc{True} if $u$ and $v$ are in $\varepsilon$-agreement in $G^+$, and \textsc{False} otherwise. Note that we do not need to store all of these values explicitly and we just introduced them for simplicity and clearance in expressing the algorithms.

\paragraph{High-level Description of the Algorithm.}The proposed online algorithm is fully described in Algorithm \ref{alg:online:algorithm}. At a high-level, we have three phases:
\begin{enumerate}
	\item \emph{Initial clustering}, where an initial graph $G_0$ is clustered into $\mathcal{C}_0$ by running the offline algorithm \textsc{CorrelationClustering}($G_0,\varepsilon$) for the given parameter $\varepsilon$. 
	\item During the \emph{graph maintenance} phase, each $\op_t$, for $t=1,2,\ldots$ is applied to $G_{t-1}$ and the graph $G_t$ is computed. This can be efficiently done by calling the corresponding algorithms, i.e. Algorithm \ref{alg:flip:sign} if $\op_t$ is flipping the sign of an edge, \textsc{AddVertex} and \textsc{DeleteVertex} if $\op_t$ is adding a new or deleting an existing vertex, respectively. 
	\item After the graph is updated, the clustering $\mathcal{C}_{t-1}$ would be refined to calculate $\mathcal{C}_t$. This phase is called \emph{cluster maintenance}. 
\end{enumerate}

\begin{algorithm}[ht]
	\caption{The \textsc{OnlineCorrelationClustering} algorithm}
	\label{alg:online:algorithm}
	\begin{algorithmic}[1]
		\Procedure{OnlineCorrelationClustering}{$G_0, \varepsilon$} 
		\State \textbf{Initialization}
		\State $\mathcal{C}_0 \gets \textsc{CorrelationClustering}(G_0, \varepsilon)$ \Comment{Algorithm \ref{alg:correlation:clustering}}
		\State \textbf{Run for each $\op_t$ for $t=1,2,\ldots$}
		\If{$\op_t$ is $\textsc{FlipSign}(e)$}
		\State Call $\textsc{FlipSign}(e)$ \Comment{Algorithm \ref{alg:flip:sign}}
		\ElsIf{$\op_t$ is $\textsc{AddVertex}(v)$}
		\State Call $\textsc{AddVertex}(v)$ \Comment{Algorithm described on Page \pageref{sec:vertex:addition:removal}}%
		\ElsIf{$\op_t$ is $\textsc{DeleteVertex}(v)$}
		\State Call $\textsc{DeleteVertex}(v)$ \Comment{Algorithm described on Page \pageref{sec:vertex:addition:removal}}%
		\Else
		\State \Return \textbf{Invalid Operation} \Comment{Ignoring $\op_t$ as it is invalid}
		\EndIf
		\EndProcedure
	\end{algorithmic}
\end{algorithm}

\paragraph{Time Complexity.} The initial clustering phase requires $\bigoh{\card{V_0} + \card{E^+_0} + D_0}$ time where $$D_0=\sum_{e=\set{u,v}\in E^+_0}\left(\deg_{G^+_0}(u) + \deg_{G^+_0}(v)\right),$$ (Theorem \ref{thm:offline:corr:clustering:complexity:time:space}). For each $\op_t$, for $t=1,2,\ldots$, the time complexity of the graph maintenance phase depends on the operation performed: (1) Adding/deleting a vertex by a call to \textsc{AddVertex}/\textsc{DeleteVertex} can be accomplished in time $\bigoh{1}$ (Theorem \ref{thm:result:vertex:add:del:complexity}). (2) Flipping the sign of a vertex by a call to $\textsc{FlipSign}(e=\set{u,v})$ is of time complexity $\bigoh{\card{V_t}+\card{E^+_t}+\left(\deg_{G^+_t}(u) + \deg_{G^+_t}(v)\right)^2}$ (Theorem \ref{thm:result:flip:sign:complexity}).

\paragraph{Space Complexity.} As the initial clustering is applied offline, its memory of $\bigoh{\card{V_t}}$ (Theorem \ref{thm:offline:corr:clustering:complexity:time:space}) does not count as the working memory required in the online phase. Combining Theorems \ref{thm:result:flip:sign:complexity} and \ref{thm:result:vertex:add:del:complexity}, the proposed algorithm requires constant extra memory to handle vertex addition and deletion, and $\bigoh{\left(\deg_{G^+_t}(u) + \deg_{G^+_t}(v)\right)^2}$ extra memory to handle edge sign flipping. 

\paragraph{Baseline Algorithm.}\label{baseline} The baseline algorithm, denoted as \textsc{Baseline}, is simply running the offline \textsc{CorrelationClustering}, Algorithm \ref{alg:correlation:clustering}, for every $G_t$ for $t=1,2,\ldots$. The time complexity of running \textsc{CorrelationClustering} for one time is as follows: (1) Constructing the graph $G^+$ can be accomplished in $\bigoh{\card{E^+}}$. (2) Discarding all edges of $G^+$ whose endpoints are not in $\varepsilon$-agreement can be done by $\bigoh{\card{E^+}}$ evaluation of $\textsc{NonAgreement}$ for two vertices. (3) Deciding on the lightness of vertices can be accomplished in $\bigoh{\card{V}}$ and the edges whose both endpoints are light can be discarded again in $\bigoh{\card{E^+}}$. (4) Finally, finding the connected components of the sparsified graph can be implemented in time $\bigoh{\card{V} + \card{E^+}}$. To sum it up, each call to $\textsc{CorrelationClustering}(G_t,\varepsilon)$ is of time $\bigoh{\card{V_t} + \card{E^+_t} + D_t}$ where $D_t=\sum_{e=\set{u,v}\in E^+}\left(\deg_{G^+_t}(u) + \deg_{G^+_t}(v)\right)$ and $G_t=(V_t,E_t)$. Moreover, the baseline algorithm requires $\bigoh{\card{V_t}}$ working memory in addition to the memory required for representing the graph $G_t$: Computing the $\nonagreement{G^+}{u}{v}$ for each edge $e=\set{u,v}\in E^+_t$ can be accomplished in $\bigoh{\deg_{G^+_t}(u)+\deg_{G^+_t}(v)}$ space, which is linear in the number of vertices in $G_t$ and can be reused for each edge. There is no need to store the \textsc{NonAgreement} of an edge, however, it is required to store the degree of each vertex $v$ in $G^+_t$ as well as the number of adjacent vertices in $G^+_t$ which are in \epsagree with $v$. Finally, representing the clustering $\mathcal{C}_t$ requires $\bigoh{\card{V_t}}$ space. 
\begin{thm}
	\label{thm:offline:corr:clustering:complexity:time:space}
	The offline $\textsc{CorrelationClustering}(G_t,\varepsilon)$, Algorithm \ref{alg:correlation:clustering}, is of time complexity $\bigoh{\card{V_t} + \card{E^+_t} + D}$ where $$D=\sum_{e=\set{u,v}\in E^+_t}\left(\deg_{G^+_t}(u) + \deg_{G^+_t}(v)\right),$$ and requires $\bigoh{\card{V_t}}$ memory plus the memory required to represent the graph $G_t$. 
\end{thm}

\paragraph{Analysis.}
Let $\mathcal{C}_t$ be the output clustering of the \textsc{Baseline} algorithm with $G_t$ as input. Also, let $\mathcal{C}^{'}_t$ be the clustering given by running \textsc{OnlineCorrelationClustering} after $\op_t$. We claim that $\mathcal{C}_t=\mathcal{C}^{'}_t$. As discussed in Sections \ref{sec:graph:maintenance:flip:sign} and \ref{sec:vertex:addition:removal}, the graph $G^+_t$ maintained by the procedures \textsc{FlipSign}, \textsc{AddVertex}, and \textsc{DeleteVertex} is the same as the graph $G^+_t$ computed in \textsc{Baseline} with the same values of $\varepsilon-\textsc{NonAgreement}$ and the same set of $\varepsilon$-lightness vertices. We have also discussed that the clustering maintenance methods described in Sections \ref{sec:clustering:maintenance} and \ref{sec:vertex:addition:removal} give the same output as the \textsc{Baseline} there. 

\subsection{Edge Sign Flip}\label{sec:edge:sign:flip}
	Let $e=\set{u,v}$ be an existing edge in the graph $G$. Then, the procedure $\textsc{FlipSign}(e)$, Algorithm \ref{alg:flip:sign}, changes the sign of the edge $e$ and maintains the clustering. Let use $H$ to denote the resulting graph after flipping the sign of the edge $e$. By combining Theorems \ref{thm:flip:edge:time:complexity:updating:g:plus} and \ref{thm:cluster:maintain:flip:sign}, we can state the following theorem.
	\begin{thm}
		\label{thm:result:flip:sign:complexity}
		The procedure $\textsc{FlipSign}(e=\set{u,v})$ stated in Algorithm \ref{alg:flip:sign} flips the sign of the edge $e$ and maintains the graph and the clustering with $$\bigoh{\card{V_t}+\card{E^+_t}+\left(\deg_{G^+_t}(u) + \deg_{G^+_t}(v)\right)^2},$$ time using $\bigoh{\left(\deg_{G^+_t}(u) + \deg_{G^+_t}(v)\right)^2}$ working memory in addition to the memory required for representing clustering output, e.g. $\mathcal{C}_t$.
	\end{thm}
	\subsubsection{Flipping from $-$ to $+$}\label{sec:edge:sign:flip:n:to:p}
		Suppose that $\textsc{FlipSign}(e=\set{u,v})$ flips the sign of an edge $e$ from $-$ to $+$. As it is illustrated in Figure \ref{fig.signflip}, this change will have two different effects on the $\varepsilon$-agreement of neighboring vertices. If a vertex $w$ belongs to the neighborhood of both $u$ and $v$, the effect of edge sign flipping would be different compared to the case where $w$ belongs to either neighborhood and not the other. 
		\begin{figure}[ht]
			\centering
			\includegraphics[width=.9\textwidth]{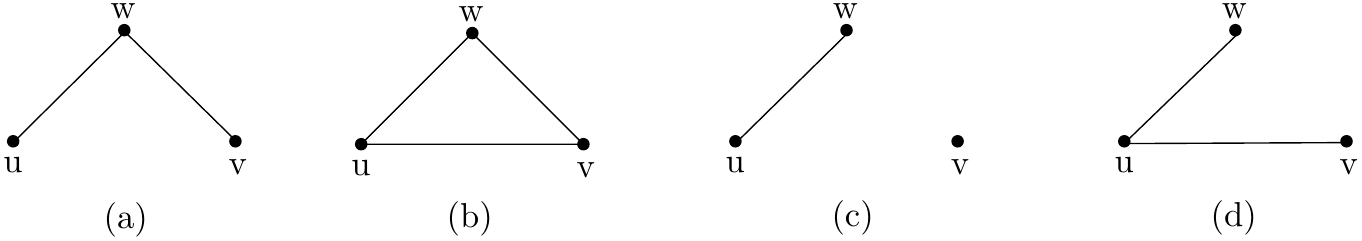}
			\caption{The effect of calling $\textsc{FlipSign}(e=\set{u,v})$ on the graph $G$ where the sign of edge $e$ is flipped from $-$ to $+$: (a) $G^+$ (b) $H^+$ (c) $G^+$ (d) $H^+$.}
			\label{fig.signflip}
		\end{figure}

		\begin{lem}
			\label{lem:case:n:to:p:eps:agreement}
			Suppose the sign of edge $u=\set{u,v}$ is changed from $-$ to $+$ by a single call to $\textsc{FlipSign}_G(e)$. Then, 
			\begin{enumerate}
				\item If $w \in N_{G^+}(u) \cap N_{G^+}(v)$, then $\nonagreement{G^+}{u}{w} \geq \nonagreement{H^+}{u}{w}$.
				\item If $w \in N_{G^+}(u) \setminus N_{G^+}(v)$, $\card{N_{G^+}(u)} < \card{N_{G^+}(w)}$, and $w \neq v$, then $\nonagreement{G^+}{u}{w} < \nonagreement{H^+}{u}{w}$.
				\item Let $\textsc{Threshold}=\card{N_{G^+}(u) \Delta N_{G^+}(w)}$. If $w \in N_{G^+}(u) \setminus N_{G^+}(v)$, $\card{N_{G^+}(u)} \geq \card{N_{G^+}(w)}$, and $w \neq v$. Then,
					\begin{enumerate}
						\item $\textsc{Threshold} < \card{N_{G^+}(u)}$ implies $\nonagreement{G^+}{u}{w} < \nonagreement{H^+}{u}{w}$.
						\item $\textsc{Threshold} = \card{N_{G^+}(u)}$ implies $\nonagreement{G^+}{u}{w} = \nonagreement{H^+}{u}{w}$. 
						\item $\textsc{Threshold} > \card{N_{G^+}(u)}$ implies $\nonagreement{G^+}{u}{w} > \nonagreement{H^+}{u}{w}$.
					\end{enumerate}
				\item Let $w \notin N_{G^+}(u) \cup N_{G^+}(v)$, then $\nonagreement{G^+}{x}{w} = \nonagreement{H^+}{x}{w}$ for $x \in N_{G^+}(w)$.
			\end{enumerate}
		\end{lem} %
		At first, it seems that maintaining the graph $G^t$ and the status of \epsagree after an edge sign flipping may require a large amount of processing. However, as we show in the following two propositions, flipping the sign of a vertex has local effects. 
		\begin{prop}
			\label{prop:eps:agreement:status:n:to:p:neigh}
			Suppose the sign of an edge $u=\set{u,v}$ is changed from $-$ to $+$ by a single call to $\textsc{FlipSign}_G(e)$ and $H$ is the resulting graph. Let $w \in N_{G+}(u)$ and $\textsc{Threshold}=\card{N_{G^+}(u) \Delta N_{G^+}(w)}$. %
			\begin{enumerate}
				\item If $u$ and $w$ are in $\varepsilon$-agreement in $G^+$, then they are in $\varepsilon$-agreement in $H^+$ if any of the following cases hold:
					\begin{enumerate}
						\item $w \in N_{G^+}(u) \cap N_{G^+}(v)$.
						\item $w \in N_{G^+}(u) \setminus N_{G^+}(v)$, $w \neq v$, $\card{N_{G^+}(u)} \geq \card{N_{G^+}(w)}$, and $\textsc{Threshold} \geq \card{N_{G^+}(u)}$.
					\end{enumerate}
				\item If $u$ and $w$ are not in $\varepsilon$-agreement in $G^+$, then they are not in $\varepsilon$-agreement in $H^+$ if any of the following cases hold:
					\begin{enumerate}
						\item $w \in N_{G^+}(u) \setminus N_{G^+}(v)$, $w \neq v$, and $\card{N_{G^+}(u)} < \card{N_{G^+}(w)}$.
						\item $w \in N_{G^+}(u) \setminus N_{G^+}(v)$, $w \neq v$, $\card{N_{G^+}(u)} \geq \card{N_{G^+}(w)}$, and $\textsc{Threshold} < \card{N_{G^+}(u)}$.
					\end{enumerate}
			\end{enumerate}
		\end{prop} %
		\begin{prop}
			\label{prop:eps:agreement:status:n:to:p:rest}
			Suppose the sign of an edge $u=\set{u,v}$ is changed from $-$ to $+$ by a single call to $\textsc{FlipSign}_G(e)$. 
			Let $w \notin N_{G^+}(u) \cup N_{G^+}(v)$, then the $\varepsilon$-agreement of $x$ and $w$ is the same in both $G^+$ and $H^+$ for $x \in N_{G^+}(w)$.
		\end{prop} %
		Propositions \ref{prop:eps:agreement:status:n:to:p:neigh} and \ref{prop:eps:agreement:status:n:to:p:rest} and Lemma \ref{lem:case:n:to:p:eps:agreement} cover all possible configurations for $\varepsilon$-agreement of two vertices $u$ and $v$ whose positive sign edge $e=\set{u,v}$ is present in $G^+$. These results are summarized in Table \ref{tbl:edge:agreement:summary}. 
		\begin{table}[ht]
			\caption{All different configurations for $\varepsilon$-agreement of any two vertices whose edge is present in $H^+$ after $\textsc{FlipEdge}(u,v)$ is called. The marks $\checkmark$, and $\times$ are used to denote that the vertices in that row are in $\varepsilon$-agreement or not. The mark $?$ is used to denote that identifying $\varepsilon$-agreement requires re-computation of the Eq. \eqref{eq:eps:agreement}.}
			\label{tbl:edge:agreement:summary}
			\centering 
			\begin{tabular}{|c|c|c|c|l|}
				\hline
				\textbf{Edge} & \textbf{Neighborhood} & \textbf{$\varepsilon$-agreement $G^+$} & \textbf{$\varepsilon$-agreement $H^+$} & \textbf{Exact Conditions}  \\ \hline \hline 
				\multicolumn{5}{|c|}{Either flipping sign from $+$ to $-$ or $-$ to $+$} \\ \hline 
				$\set{x,w}$ & \multirow{2}{*}{$\substack{w \notin N_{G^+}(u) \cup N_{G^+}(v)\\w\notin \set{u,v}}$} & $\checkmark$ & $\checkmark$ & \multirow{2}{*}{Propositions \ref{prop:eps:agreement:status:n:to:p:rest} \& \ref{prop:eps:agreement:status:p:to:n:rest}} \\ %
				$\set{x,w}$ &  & $\times$ & $\times$ & \\ \hline \hline 
				\multicolumn{5}{|c|}{Flipping sign from $-$ to $+$} \\ \hline 
				$\set{u,w}$ & \multirow{2}{*}{$w\in N_{G^+}(u) \cap N_{G^+}(v)$} & $\checkmark$ & $\checkmark$ & Proposition \ref{prop:eps:agreement:status:n:to:p:neigh}, (1-a) \\ %
				$\set{u,w}$ &  & $\times$ & $?$ & Lemma \ref{lem:case:n:to:p:eps:agreement}, (1) \\ \hline %
				$\set{u,w}$ & \multirow{4}{*}{$w\in N_{G^+}(u) \setminus N_{G^+}(v)$} & $\checkmark$ & $\checkmark$ & Proposition \ref{prop:eps:agreement:status:n:to:p:neigh}, (1-b) \\ %
				$\set{u,w}$ &  & $\checkmark$ & $?$ & Lemma \ref{lem:case:n:to:p:eps:agreement}, (2,3-a) \\ %
				$\set{u,w}$ &  & $\times$ & $\times$ & Proposition \ref{prop:eps:agreement:status:n:to:p:neigh}, (2-a,2-b) \\ %
				$\set{u,w}$ &  & $\times$ & $?$ & Lemma \ref{lem:case:n:to:p:eps:agreement}, (3-c) \\ \hline 
				\hline 
				\multicolumn{5}{|c|}{Flipping sign from $+$ to $-$} \\ \hline 
				$\set{u,w}$ & \multirow{2}{*}{$w\in N_{G^+}(u) \cap N_{G^+}(v)$} & $\checkmark$ & $?$ & Lemma \ref{lem:case:p:to:n:eps:agreement}, (1) \\ %
				$\set{u,w}$ &  & $\times$ & $\times$ & Proposition \ref{prop:eps:agreement:status:n:to:p:neigh}, (2-a) \\ \hline %
				$\set{u,w}$ & \multirow{4}{*}{$w\in N_{G^+}(u) \setminus N_{G^+}(v)$} & $\checkmark$ & $\checkmark$ & Proposition \ref{prop:eps:agreement:status:p:to:n:neigh}, (1) \\ %
				$\set{u,w}$ &  & $\checkmark$ & $?$ & Lemma \ref{lem:case:p:to:n:eps:agreement}, (2,3-c) \\ %
				$\set{u,w}$ &  & $\times$ & $\times$ & Proposition \ref{prop:eps:agreement:status:p:to:n:neigh}, (2-b) \\ %
				$\set{u,w}$ &  & $\times$ & $?$ & Lemma \ref{lem:case:p:to:n:eps:agreement}, (3-a) \\ \hline 
			\end{tabular}
		\end{table}
	\subsubsection{Flipping from $+$ to $-$}\label{sec:edge:sign:flip:p:to:n}
		Similar to the flipping the sign of an edge from $-$ to $+$, we can obtain a similar set of results, however with different conditions. 
		\begin{lem}
			\label{lem:case:p:to:n:eps:agreement}
			Suppose the sign of an edge $u=\set{u,v}$ is changed from $+$ to $-$ by a single call to $\textsc{FlipSign}_G(e)$. Then,
			\begin{enumerate}
				\item If $w \in N_{G^+}(u) \cap N_{G^+}(v)$, then $\nonagreement{H^+}{u}{w} > \nonagreement{G^+}{u}{w}$.
				\item If $w \in N_{G^+}(u) \setminus N_{G^+}(v)$, $\card{N_{G^+}(u)} \leq \card{N_{G^+}(w)}$, and $w \neq v$, then $\nonagreement{H^+}{u}{w} < \nonagreement{G^+}{u}{w}$.
				\item Let $\textsc{Threshold}=\card{N_{G^+}(u) \Delta N_{G^+}(w)}$. If $w \in N_{G^+}(u) \setminus N_{G^+}(v)$, $\card{N_{G^+}(u)} > \card{N_{G^+}(w)}$, and $w \neq v$. Then,
				\begin{enumerate}
					\item $\textsc{Threshold} < \card{N_{G^+}(u)}$ implies $\nonagreement{G^+}{u}{w} > \nonagreement{H^+}{u}{w}$.
					\item $\textsc{Threshold} = \card{N_{G^+}(u)}$ implies $\nonagreement{G^+}{u}{w} = \nonagreement{H^+}{u}{w}$. 
					\item $\textsc{Threshold} > \card{N_{G^+}(u)}$ implies $\nonagreement{G^+}{u}{w} < \nonagreement{H^+}{u}{w}$.
				\end{enumerate}
				\item Let $w \notin N_{G^+}(u) \cup N_{G^+}(v)$, then $\nonagreement{G^+}{x}{w} = \nonagreement{H^+}{x}{w}$ for $x \in N_{G^+}(w)$.
			\end{enumerate}
		\end{lem} %
		As we will show in the following two propositions, the effect of edge sign flipping from $+$ to $-$ is local similar to flipping the edge sign from $-$ to $+$, although has different effects. 
		\begin{prop}
			\label{prop:eps:agreement:status:p:to:n:neigh}
			Suppose the sign of an edge $u=\set{u,v}$ is changed from $+$ to $-$ by a single call to $\textsc{FlipSign}_G(e)$. Let $w \in N_{G+}(u)$ and $\textsc{Threshold}=\card{N_{G^+}(u) \Delta N_{G^+}(w)}$. %
			\begin{enumerate}
				\item If $u$ and $w$ are in $\varepsilon$-agreement in $G^+$, then they are in $\varepsilon$-agreement in $H^+$ if any of the following cases hold:
				\begin{enumerate}
					\item $w \in N_{G^+}(u) \setminus N_{G^+}(v)$, $w \neq v$, and $\card{N_{G^+}(u)} \leq \card{N_{G^+}(w)}$.
					\item $w \in N_{G^+}(u) \setminus N_{G^+}(v)$, $w \neq v$, $\card{N_{G^+}(u)} > \card{N_{G^+}(w)}$, and $\textsc{Threshold} \leq \card{N_{G^+}(u)}$.
				\end{enumerate}
				\item If $u$ and $w$ are not in $\varepsilon$-agreement in $G^+$, then they are not in $\varepsilon$-agreement in $H^+$ if any of the following cases hold:
				\begin{enumerate}
					\item $w \in N_{G^+}(u) \cap N_{G^+}(v)$.
					\item $w \in N_{G^+}(u) \setminus N_{G^+}(v)$, $w \neq v$, $\card{N_{G^+}(u)} > \card{N_{G^+}(w)}$, and $\textsc{Threshold} > \card{N_{G^+}(u)}$.
				\end{enumerate}
			\end{enumerate}
		\end{prop} %
		\begin{prop}
			\label{prop:eps:agreement:status:p:to:n:rest}
			Suppose the sign of an edge $u=\set{u,v}$ is changed from $+$ to $-$ by a single call to $\textsc{FlipSign}_G(e)$. Let $w \notin N_{G^+}(u) \cup N_{G^+}(v)$, then the $\varepsilon$-agreement of $x$ and $w$ is the same in both $G^+$ and $H^+$ for $x \in N_{G^+}(w)$.
		\end{prop} %
	\subsubsection{Maintenance Algorithm for $G^+$}\label{sec:graph:maintenance:flip:sign}
		After building our basic tools in Sections \ref{sec:edge:sign:flip:p:to:n} and \ref{sec:edge:sign:flip:n:to:p}, we are ready to state an algorithm for maintaining the graph $G^+$. We can benefit from the locality of the effect of the \textsc{FlipSign}, as we show in the following corollary. 
		\begin{coro}
			\label{cor:agreecount:not:changing:non:neighbors}
			Suppose the sign of an edge $u=\set{u,v}$ is changed by a single call to $\textsc{FlipSign}_G(e)$. Then, $\agreecnt_{G^+}(w)=\agreecnt_{H^+}(w)$ for all $w\in V(G) \setminus S$ where $S=N_{G^+}(u)\cup N_{G^+}(v)\cup \set{u,v}$. 
		\end{coro} %
		In other words, Corollary \ref{cor:agreecount:not:changing:non:neighbors} states that after a call to $\textsc{FlipSign}_G(e)$, we just need to take care of the $\agreecnt_{H^+}(w)$ in the subgraph induced by set $S$, e.g. $w\in S$. 
		
		The neighborhood of a vertex $w\in V$ in $G^+$, i.e. $N_{G^+}(w)$, can be partitioned into the following 4 parts: (1) $A=N_{G^+}(w) \cap N_{G^+}(u) \cap N_{G^+}(v)$, (2) $B=N_{G^+}(w) \cap \left(N_{G^+}(u) \setminus N_{G^+}(v)\right)$, (3) $C=N_{G^+}(w) \cap \left(N_{G^+}(v) \setminus N_{G^+}(u)\right)$, and (4) $D=N_{G^+}(w) \setminus \left(N_{G^+}(u) \cup N_{G^+}(v)\right)$. 
		If the vertex $w$ and a vertex in $D$ are in \epsagree in $G^+$, then, they would be in \epsagree in $H^+$, too. This is implied by the Proposition \ref{prop:eps:agreement:status:n:to:p:neigh} (the second case) and the Proposition \ref{prop:eps:agreement:status:p:to:n:rest}, i.e. the \epsagree of $w$ and the vertices in the set $D$ is equal in both $G^+$ and $H^+$. However, for the other three parts, $A$, $B$ and $C$, we need to consider the exact flip: %
		\begin{quote}\begin{description}
			\item[$-$ to $+$:] The \epsagree between the vertex $w$ and the vertices in either sets $B$ and $C$ requires verification, depending on extra conditions (Table \ref{tbl:edge:agreement:summary}, rows 3-10). However, the \epsagree of the vertex $w$ with vertices of $A$ in $G^+$ implies their \epsagree in $H^+$ by the Proposition \ref{prop:eps:agreement:status:n:to:p:rest}. 
			\item[$+$ to $-$:] By the Table \ref{tbl:edge:agreement:summary}, the \epsagree between the vertex $w$ and the vertices in either sets $A$, $B$ and $C$ requires verification, except for the vertices in $A$ which are not in \epsagree with $w$ in $G^+$, remain in non-\epsagree with $w$ in $H^+$. 
		\end{description}\end{quote}
		Considering this discussion, we can state Theorem \ref{thm:lightness:edge:flip}.
		\begin{thm}
			\label{thm:lightness:edge:flip}
			Suppose the sign of edge $f=\set{u,v}$ is changed by a single call to $\textsc{FlipSign}_G(e)$ and Let $S=N_{G^+}(u)\cup N_{G^+}(v)\cup \set{u,v}$. 
			\begin{enumerate}
				\item For all vertices $w \notin S$, if $w$ is $\varepsilon$-light in $G^+$, then it would remain $\varepsilon$-light in $H^+$, too. 
				\item For all vertices $w \notin S$, if $w$ is not a $\varepsilon$-light vertex in $G^+$, then it would not be $\varepsilon$-light in $H^+$, also.
				\item The $\varepsilon$-lightness of vertices in $S$ requires re-computation and verification in $H^+$. 
			\end{enumerate}
		\end{thm}
		
		\paragraph{Algorithm Intuition.} The major role of the Algorithms \ref{alg:update:g:plus:flip:sign:n:to:p} and \ref{alg:update:g:plus:flip:sign:p:to:n} is keeping the lightness of vertices valid during the \textsc{FlipSign} updates. The naive way of doing this requires a re-computation of all $\varepsilon$-agreements throughout the updated graph. However, using Corollary \ref{cor:agreecount:not:changing:non:neighbors} and Theorem \ref{thm:lightness:edge:flip}, we know that the lightness of vertices, except the ones in the neighborhood of the edge endpoints, does not change in a graph after a single call to \textsc{FlipSign}. In other words, the changes are only in the subgraph with vertex set $S = \set{u,v} \cup N_{G^+}(u) \cup N_{G^+}(v)$. Therefore, both the Algorithms \ref{alg:update:g:plus:flip:sign:n:to:p} and \ref{alg:update:g:plus:flip:sign:p:to:n} need to recompute the number of vertices in $\varepsilon$-agreement with the set $S$. To accomplish this efficiently, these Algorithms use Table \ref{tbl:edge:agreement:summary} and apply a pruning strategy. Note that while verifying the $\varepsilon$-agreement of two vertices by calling \textsc{VerifyEdge} (Algorithm \ref{alg:verify:edge:sign:flip}), the number of vertices in \epsagree with the edge endpoints is updated. This is later used to update the lightness of vertices in $S$. 
		
		After updating the graph $G^+$ to accommodate the flipped edge, we need to update the underlying clusters. We will defer this discussion to Section \ref{sec:clustering:maintenance}. However, for now we just call the $\textsc{MaintainClusteringAfterFlipSign}$ algorithm. 
		\paragraph{Time Complexity.} It is easy to verify that evaluating the $\nonagreement{G^+}{u}{v}$ requires $$\bigoh{\deg_{G^+}(u) + \deg_{G^+}(v)}=\bigoh{\Delta_{G^+}},$$ operations at the worst-case, where $\Delta_{G^+}$ is the maximum degree of the graph $G^+$. Therefore, calling \textsc{VerifyEdge}($f=\set{u,v})$ is of $\bigoh{\deg_{G^+}(u) + \deg_{G^+}(v)}$ or $\bigoh{\Delta_{G^+}}$ time complexity, in the worst case. 
		
		Lines 4-6 in $\textsc{Update}G^+\textsc{NegativeToPositive}$ (Algorithm \ref{alg:update:g:plus:flip:sign:n:to:p}) can be accomplished in time $\bigoh{\deg_{G^+}(u) + \deg_{G^+}(v)}$. The \texttt{For} loop on lines 7-15 calls the \textsc{VerifyEdge} procedure for at most $\bigoh{\card{A}}$ times. Moreover, the \texttt{For} loop on lines 18-32 calls the \textsc{VerifyEdge} procedure for at most $\bigoh{\card{B}}$ and $\bigoh{\card{C}}$ times for $X=B$ and $X=C$, respectively. As $A$, $B$ and $C$ are a partition of $N_{G^+}(u) \cup N_{G^+}(v)$, the procedure $\textsc{Update}G^+\textsc{NegativeToPositive}$ calls \textsc{VerifyEdge} for at most $\bigoh{\deg_{G^+}(u) + \deg_{G^+}(v)}$ times. Lines 35-41 in $\textsc{Update}G^+\textsc{NegativeToPositive}$ requires an evaluation of $\nonagreement{G^+}{u}{v}$, which can be accomplished in $\bigoh{\deg_{G^+}(u) + \deg_{G^+}(v)}$ time. Moreover, lines 42-45 requires $\bigoh{\card{S}}$ time where $\card{S}=\bigoh{\deg_{G^+}(u) + \deg_{G^+}(v)}$. As the same analysis applies to $\textsc{Update}G^+\textsc{PositiveToNegative}$ (Algorithm \ref{alg:update:g:plus:flip:sign:p:to:n}), we can state Theorem \ref{thm:flip:edge:time:complexity:updating:g:plus}.
		
		\paragraph{Space Complexity.} A careful examination of the $\textsc{FlipSign}(e=\set{u,v})$ algorithm shows that this algorithm requires $\bigoh{\deg_{G^+_t}(u)+\deg_{G^+_t}(v)}$ working memory in addition to the memory required for representing the graph $G^+[S]$ where $S=N_{G^+_t}(u) \cup N_{G^+_t}(v) \cup \set{u,v}$: (1) There is no need to construct the graph $G^+$, as we can simply substitute $G$ with $G^+$, every edge not present in $G^+$ has a negative sign. (2) Moreover, we do not need to load the complete graph $G^+$ into memory, thanks to the locality of edge flipping (Theorem \ref{thm:lightness:edge:flip}), we just need to load $G^+[S]$. (3) We need to store the set $S$, which in the worst-case can require $\bigoh{\card{V}}$ memory (consider a complete graph with all positive edges). As this worst-case scenario is very pessimistic, we use the size of the set $S$, i.e. $\bigoh{\deg_{G^+_t}(u)+\deg_{G^+_t}(v)}$ memory. (4) Verifying the \epsagree between two vertices also requires $\bigoh{\deg_{G^+_t}(u)+\deg_{G^+_t}(v)}$ memory which is reused for every edge in $G^+[S]$. Note that the time and space complexity stated in Theorem \ref{thm:flip:edge:time:complexity:updating:g:plus} only applies to maintaining the graph $G^+_t$. The time and space required for maintaining the clustering is stated in Theorem \ref{thm:cluster:maintain:flip:sign}. 
		\begin{thm}
			\label{thm:flip:edge:time:complexity:updating:g:plus}
			Suppose the sign of an edge $f=\set{u,v}$ is changed by a single call to $\textsc{FlipSign}_G(e)$. Then, updating the $\varepsilon$-lightness of vertices in $G$ can be accomplished in $\bigoh{\left(\deg_{G^+}(u) + \deg_{G^+}(v)\right)^2}$ time and $\bigoh{\deg_{G^+}(u)+\deg_{G^+}(v)}$ working memory in addition to the memory required for representing $G^+[S]$ where $S=N_{G^+}(u) \cup N_{G^+}(v) \cup \set{u,v}$.
		\end{thm} %
		\begin{algorithm}[ht]
			\caption{$\textsc{Update}G^+\textsc{NegativeToPositive}$ procedure runs after every call to $\textsc{FlipSign}_G(e)$ which changed $e.\textsc{Sign}$ from $-$ to $+$.}
			\label{alg:update:g:plus:flip:sign:n:to:p}
			\begin{algorithmic}[1]
				\Procedure{Update$G^+\textsc{NegativeToPositive}$}{$e=\set{u,v},i$}
					\State Let $G_t$ be the graph after the $t^{th}$ modification command which was $\textsc{FlipSign}_G(e)$. 
					\State $G^+ \gets G_t[E^+]$
					\State Let $A \gets N_{G^+}(u) \cap N_{G^+}(v)$
					\State Let $B=\left(N_{G^+}(u) \setminus N_{G^+}(v)\right) \setminus \set{v}$
					\State Let $C=\left(N_{G^+}(v) \setminus N_{G^+}(u)\right) \setminus \set{u}$
					\ForAll{$w \in A$} \Comment{Applying Lemma \ref{lem:case:n:to:p:eps:agreement}, (1)}
						\State Let $f=\set{u,w}$ and $g=\set{v,w}$
						\If{\textbf{not} $f.\textsc{Agree}$}
							\State Call \textsc{VerifyEdge}($f=\set{u,w}$)
						\EndIf
						\If{\textbf{not} $g.\textsc{Agree}$}
							\State Call \textsc{VerifyEdge}($g=\set{v,w}$)
						\EndIf
					\EndFor
					\For{$(x, X)$ in $[(u, B), (v, C)]$} 
						\State \Comment{Run once with $x\gets u$ and $X\gets B$, another time with $x\gets v$ and $X\gets C$}
						\ForAll{$w \in X$}
							\State Let $f=\set{x,w}$
							\If{$f.\textsc{Agree}$}
								\If{$\deg_{G^+}(x) < \deg_{G^+}(w)$} \Comment{Applying Lemma \ref{lem:case:n:to:p:eps:agreement}, (2)}
									\State Call \textsc{VerifyEdge}($f=\set{x,w}$)
								\ElsIf{$\card{N_{G^+}(x) \Delta N_{G^+}(w)} < \deg_{G^+}(x)$} \Comment{Applying Lemma \ref{lem:case:n:to:p:eps:agreement}, (3-a)}
									\State Call \textsc{VerifyEdge}($f=\set{x,w}$)
								\EndIf
							\Else \Comment{$f.\textsc{Agree} = \textsc{False}$}
								\If{$\deg_{G^+}(x) \geq \deg_{G^+}(w)$ \textbf{ and } $\card{N_{G^+}(x) \Delta N_{G^+}(w)} > \deg_{G^+}(x)$} 
									\State \Comment{Applying Lemma \ref{lem:case:n:to:p:eps:agreement}, (3-c)}
									\State Call \textsc{VerifyEdge}($f=\set{x,w}$)
								\EndIf
							\EndIf
						\EndFor
					\EndFor
					\State Let $f=\set{u,v}$
					\If{$\nonagreement{G^+}{u}{v}\leq \varepsilon$}
						\State $f.\textsc{Agree} \gets \textsc{True}$
						\State $u.\agreecnt_{G^+} \gets u.\agreecnt_{G^+} + 1$
						\State $v.\agreecnt_{G^+} \gets v.\agreecnt_{G^+} + 1$
					\Else
						\State $f.\textsc{Agree} \gets \textsc{False}$
					\EndIf
					\State Let $S \gets \set{u,v} \cup u.\textsc{Neigh}_{G^+} \cup v.\textsc{Neigh}_{G^+}$
					\ForAll{$w \in S$} \Comment{Fixing the lightness of the vertices in $G^+$}
						\State $w.\textsc{IsLight} \gets \left(\frac{w.\agreecnt_{G^+}}{\deg_{G^+}(w)} < \varepsilon\right)$
					\EndFor
					\State \Return $G^+$ as $H^+$.
				\EndProcedure
			\end{algorithmic}
		\end{algorithm}
		\begin{algorithm}[h]
			\caption{The \textsc{VerifyEdge($f=\set{x,y}$)} maintains $f.\textsc{Agree}$, $x.\agreecnt$ and $y.\agreecnt$ whenever an update occurs.}
			\label{alg:verify:edge:sign:flip}
			\begin{algorithmic}[1]
				\Procedure{VerifyEdge}{$f=\set{x,y}$}
					\If{\textbf{not} $f.\textsc{Agree}$ \textbf{and} $\nonagreement{G^+}{x}{y} < \varepsilon$}
						\State $f.\textsc{Agree} \gets \textsc{True}$
						\State $x.\agreecnt \gets x.\agreecnt + 1$
						\State $y.\agreecnt \gets y.\agreecnt + 1$
					\ElsIf{$f.\textsc{Agree}$ \textbf{and} $\nonagreement{G^+}{x}{y} \geq \varepsilon$}
						\State $f.\textsc{Agree} \gets \textsc{False}$
						\State $x.\agreecnt \gets x.\agreecnt - 1$
						\State $y.\agreecnt \gets y.\agreecnt - 1$
					\EndIf
				\EndProcedure
			\end{algorithmic}
		\end{algorithm}
		\begin{algorithm}[h]
			\caption{\textsc{FlipSign}($e$) flips the sign of an edge $e$ and updates the corresponding clusters}
			\label{alg:flip:sign}
			\begin{algorithmic}[1]
				\Procedure{FlipSign}{$e$}
					\State $G_t \gets G_{t-1}$
					\State $\textsc{PreviousSign} \gets e.\sign$
					\State Change the sign of edge $e$ in $G_t$
					\State Let $G^+=G_t[E^+]$
					\If{\textsc{PreviousSign} is $+$}
						\State Call $\textsc{Update}G^+\textsc{PositiveToNegative}(e)$ \Comment{Updates $G_t$ to reflect sign flip of edge}
					\Else \Comment{\textsc{PreviousSign} is $-$}
						\State Call $\textsc{Update}G^+\textsc{NegativeToPositive}(e)$ \Comment{Updates $G_t$ to reflect sign flip of edge}
					\EndIf
					\State Call $\textsc{MaintainClusteringAfterFlipSign}$ \Comment{Described on Page \pageref{alg:description:cluster:maintain:flip:sign}}
				\EndProcedure
			\end{algorithmic}
		\end{algorithm}

\subsection{Maintaining Clusters in Edge Sign Flip}\label{sec:clustering:maintenance}
	Let $\mathcal{C}_t$ be the clustering at time step $t$ for $t=0,1,\ldots$. Also, let $\op_1,\op_2,\ldots$ be a sequence of operations where $\op_t$ is applied to the graph $G_{t-1}$ and $G_t$ is its result. 
	\begin{defn}[$\varepsilon$-critical]
		\label{def:eps:critical}
		Let $G$ be a signed graph and $G^+=G[E^+]$ and assume that $\widetilde{G^+}$ be the sparsified graph of $G^+$ (Line 5 of the Algorithm \ref{alg:correlation:clustering}). Then, a vertex $v$ (edge $e$) whose removal modifies the connected components in $\widetilde{G^+}$ is called an $\varepsilon$-critical vertex (edge).
	\end{defn}
	Note that we used the term ``modifying the connected components'', rather than talking about the change in the number of components in Definition \ref{def:eps:critical}. Assume the graph in Figure \ref{fig:example:graph} with the following setting: all the edges have positive signs, vertices $v,w,x,y$ and $z$ are $\varepsilon$-light, and vertex $u$ is not $\varepsilon$-light. Moreover, assume that any path between vertices $x$ and $y$ passes the vertex $u$, any path between vertices $w$ and $z$ passes the vertex $v$ and any path between $w$ and either $x,y$ or between $z$ and either $x,y$ passes through the edge $\set{u,v}$. In this setting, the edges $\set{v,w}$ and $\set{v,z}$ would be removed in the corresponding sparsified graph. Without loss of generality, assume that the sparsified graph consisted of 3 components. It is clear that flipping the sign of edge $\set{u,v}$ would result in its removal. However, consider removing the edge $\set{u,v}$: it makes $u$ an $\varepsilon$-light vertex, however changes the vertex $v$ to a non-$\varepsilon$-light vertex. In this case, the edges $\set{x,u}$ and $\set{y,u}$ would be removed in the corresponding sparsified graph whereas the edges $\set{v,w}$ and $\set{v,z}$ remain. In this case, the number of connected components remains unchanged, e.g. 3, but the underlying clustering is changed. 
	\begin{figure}[h]
		\centering
		\includegraphics[width=.3\textwidth]{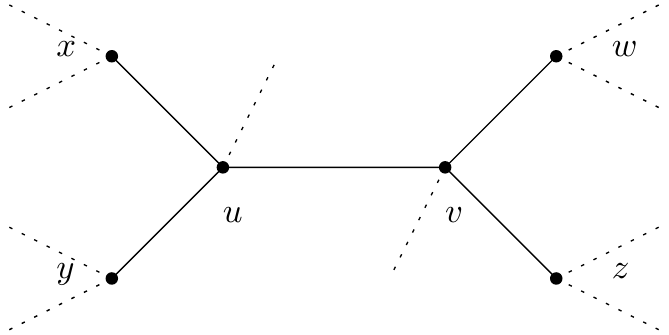}
		\caption{A $\varepsilon$-critical edge whose removal does not change the number of connected components in the sparsified graph.}
		\label{fig:example:graph}
	\end{figure}
	
	Moreover, it is notable that any cut edge of $G^+$ is a $\varepsilon$-critical edge, however, there may be an edge $e=\set{u,v}$ which is not a cut edge of $G^+$, but is $\varepsilon$-critical. For example, consider the graph $G^+$ in Figure \ref{fig:critical:edge:not:cut}(a), consisting of a single connected component, where vertices $w$ and $x$ are both $\varepsilon$-light and $u$ and $v$ are not $\varepsilon$-light. Also assume that by removing the edge $e\set{u,v}$, the vertices $u$ and $v$ would become $\varepsilon$-light. In this case, the edge $e=\set{u,v}$ is not a cut, however, if deleted, the components $C_u$, $C_v$, $C_w$ and $C_x$ will not be connected in the corresponding sparsified graph. The same holds for $\varepsilon$-critical vertices, i.e. a cut vertex of $G^+$ is always $\varepsilon$-critical, however, the converse is not necessarily true. For example, in Figure \ref{fig:critical:edge:not:cut}(b), the vertex $y$ is not a cut vertex, but assuming that its deletion causes vertices $u$ and $v$ to become an $\varepsilon$-light vertex. Prior to deleting the vertex $y$, assume that vertex $w$ was $\varepsilon$-light. Then, removal of $y$ causes removal of the edges $\set{u,w}$ and $\set{v,w}$, which causes separation of $C_w$ from the rest of the graph. 
	\begin{figure}[h]
		\centering
		\includegraphics[width=.9\textwidth]{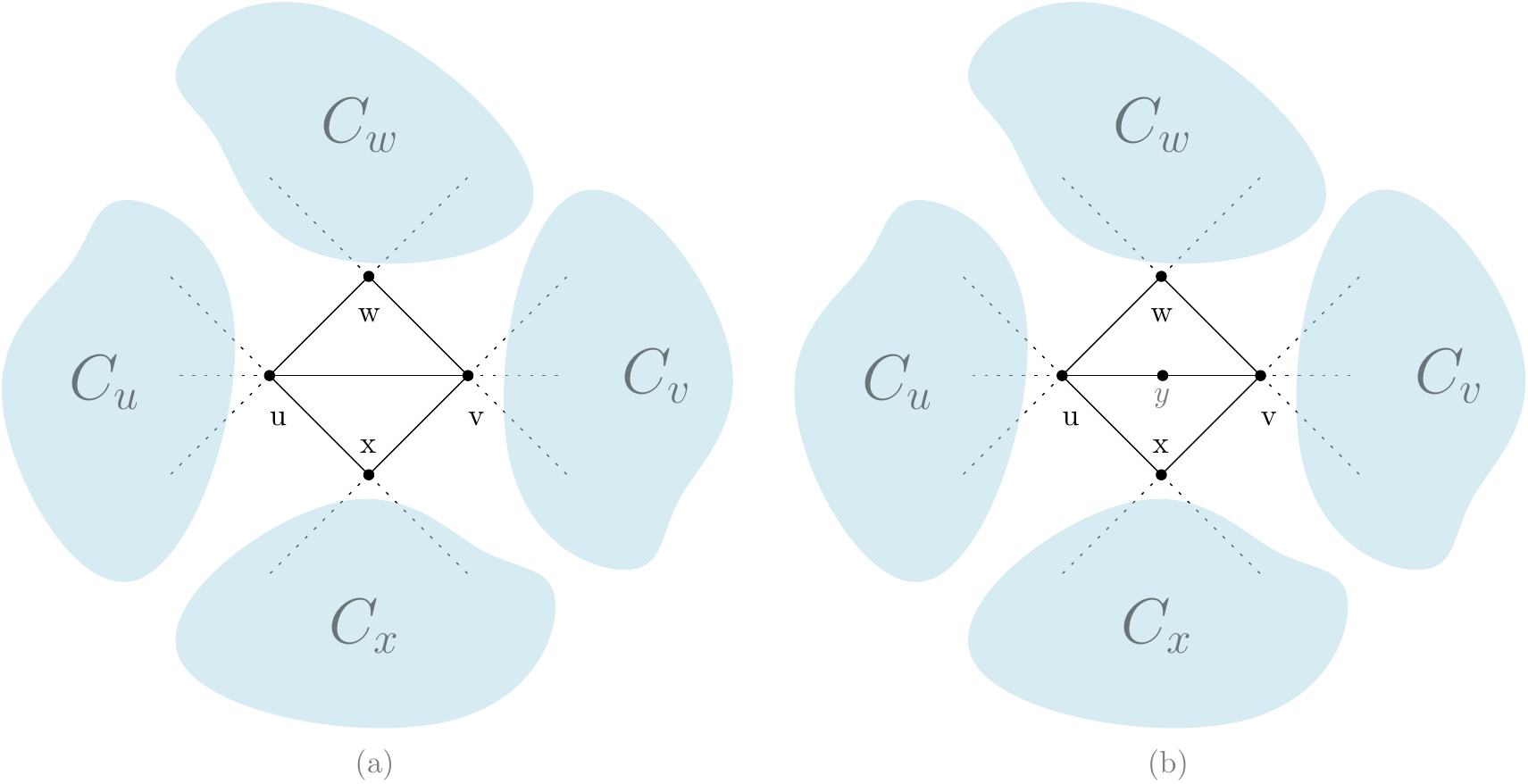}
		\caption{Counterexamples where (a) $e=\set{u,v}$ is a $\varepsilon$-critical edge, but not a cut edge of $G^+$. (b) $y$ is  a $\varepsilon$-critical vertex, but not a cut vertex of $G^+$.}
		\label{fig:critical:edge:not:cut}
	\end{figure}

	Summing up the previous discussion regarding $\varepsilon$-critical vertices and edges, we can state the following result.
	\begin{lem}
		\label{lem:flip:sign:change:sparsified:graph}
		Let $\op_t=\textsc{FlipSign}(e)$ for an edge $e=\set{u,v}$ and $S=N_{G^+_{t-1}}(u) \cup N_{G^+_{t-1}}(v) \cup \set{u,v}$.%
		Then, for all $w\notin S$, the presence of all edges $f=\set{x,w}$ with positive sign in $\widetilde{G_{t-1}^+}$ and $\widetilde{G_{t}^+}$ does not change, except for the following case where a change is possible:
		\begin{itemize}
			\item $w$ is an $\varepsilon$-light vertex in $G^+_{t-1}$, and 
			\item $x \in S$, and %
			\item $w$ and $x$ are in $\varepsilon$-agreement.
		\end{itemize}
	\end{lem} %
	Moreover, we can use the following result to prune the search space. 
	\begin{lem}
		\label{lem:flip:sign:prune:state:space}
		Let $\op_t=\textsc{FlipSign}(e)$ for an edge $e=\set{u,v}$. If there exists a cluster $C\in\mathcal{C}_{t-1}$ such that $C\subseteq V \setminus \left(N_{G^+_{t-1}}(u) \cup N_{G^+_{t-1}}(v) \cup \set{u,v}\right)$, we have $C\in\mathcal{C}_t$. 
	\end{lem} %

	\paragraph{Maintaining Clustering After an Edge Sign Flip}\label{alg:description:cluster:maintain:flip:sign}
	The basic idea in maintaining the clustering after a $\textsc{FlipSign}(e=\set{u,v})$ is updating the underlying graph $G^+_t$ and then, computing the connected components of $\widetilde{G_{t}^+}[S]$ for the set $S=N_{G^+_{t-1}}(u) \cup N_{G^+_{t-1}}(v) \cup \set{u,v}$. Let $D_1,\ldots,D_k$ be the connected components of $\widetilde{G_{t}^+}[S]$ and use $D^w$ to denote the connected component which vertex $w$ belongs. Moreover, we define a status ordering for clusters $C\in\mathcal{C}_{t-1}$ as (1) \unprocessed, (2) \copycand, (3) \mergecand, and (4) \splitcand. When marking, a mark would be applied only if the order of the new mark is strictly greater than the current mark. The procedure $\textsc{Mark}(C,\textsc{Mark},t-1)$ asks for applying \textsc{Mark} to the cluster $C$ in $\mathcal{C}_{t-1}$. The aim of these statuses and their ordering would become clear in the sequel. Now, we are ready to state the cluster maintenance algorithm after a $\textsc{FlipSign}(e=\set{u,v})$. 
	\begin{enumerate}
		\item At first, we mark all clusters $C\in\mathcal{C}_{t-1}$ as \unprocessed.
		\item For all clusters $C\in\mathcal{C}_{t-1}$ which satisfy $C\subseteq V_t\setminus S$, we mark them as \copycand. These clusters do not change and would be present in $\mathcal{C}_t$ by Lemma \ref{lem:flip:sign:prune:state:space}.
		\item For each pair of vertices $w,w' \in S$, mark $C^w$ and $C^{w'}$ as either \mergecand if $D^w=D^{w'}$ and $C^w \neq C^{w'}$, or \copycand if $D^w=D^{w'}$ and $C^w=C^{w'}$. Moreover, just for clarity and ease of notation, assume that with each \mergecand cluster, we store the set of pairs $(w,w')$ which caused $C^w$ and $C^{w'}$ to be marked as \mergecand.%
		A set of clusters marked as \mergecand with shared vertex pairs is called a cluster group, since we intend to merge them and construct a new cluster. Moreover, if $D^w\neq D^{w'}$ and $C^w = C^{w'}$, then mark $C^w$ as \splitcand. Note that when marking a cluster as \splitcand, we store its cluster group as the algorithm tries to apply \mergecand to it, although we do not change the mark. 
	\end{enumerate}
	At the end of this process, no cluster is marked \unprocessed in $\mathcal{C}_{t-1}$. Starting with an empty collection $\mathcal{C}_t$, we first add all the clusters of $C\in\mathcal{C}_{t-1}$ which are marked \copycand to $\mathcal{C}_t$ with their own ids. Then, we merge each cluster group and the resulting cluster gets a new id, but we postpone its addition to $\mathcal{C}_t$ to the time after processing clusters with \splitcand mark. Next, for each cluster $C$ with \splitcand mark, we need to compute the connected components of a modification of $\widetilde{G_{t-1}^+}[C]$ where all edges with both endpoints in $S$ are removed. If this subgraph is connected, then we merge $C$ with its corresponding cluster group and continue to the next cluster with \splitcand. However, suppose there are more than one connected components, e.g. $C_1,\ldots,C_o$. Then, for each connected component $C_i$, if $C_i \subseteq V_t \setminus S$, we add it directly to the $\mathcal{C}_t$ with a new id. Otherwise, let $w \in C \cap S$. Then, we simply merge it to the cluster group containing a pair $(w,x)$ for any $x \in S$. The new cluster replaces the old merged cluster in $\mathcal{C}_{t-1}$ and inherits its mark and cluster group.
	
	Whenever there is no more clusters marked \splitcand in $\mathcal{C}_{t-1}$, we append all the remaining clusters (which are marked \mergecand) in $\mathcal{C}_{t-1}$ to $\mathcal{C}_{t}$. We call this algorithm $\textsc{MaintainClusteringAfterFlipSign}$.
	
	\begin{thm}
		\label{thm:cluster:maintain:flip:sign}
		The algorithm $\textsc{MaintainClusteringAfterFlipSign}$ gives the same output as the \textsc{Baseline} at any time step $t=1,2,\ldots$ with $\op_t=\textsc{FlipSign}(e)$ and runs in $$\bigoh{\card{V_t}+\card{E^+_t}+\left(\deg_{G^+_t}(u) + \deg_{G^+_t}(v)\right)^2+\card{\mathcal{C}_{t-1}}},$$ time with $\bigoh{\left(\deg_{G^+_t}(u) + \deg_{G^+_t}(v)\right)^2}$ working memory in addition to memory required for representing clustering $\mathcal{C}_t$. %
	\end{thm} %

\subsection{Adding and Removing Vertices}\label{sec:vertex:addition:removal}
	Adding or removing a vertex with all negative edges is very simple, as it has no positive edges appearing in $G^+_t$. To be precise, to add a new (remove an existing) vertex to the graph $G$ with no positive incident edges, just add (remove) it. In the case of adding a new vertex $v$, this would create a new singleton cluster consisting of that single vertex, i.e. $\mathcal{C}_t=\mathcal{C}_{t-1}\cup \set{\set{v}}$. Similarly, removing an existing vertex $v$ from $G$ with no positive edges results in deleting its singleton cluster, i.e. $\mathcal{C}_t=\mathcal{C}_{t-1} \setminus \set{\set{v}}$. Both of these operations are of constant time complexity, i.e. $\bigoh{1}$. By the way, to add a new vertex with some positive incident edges, we can first add the vertex as a vertex with no positive edges and then, add its positive edges by flipping the corresponding edges. Similarly, if we need to delete a vertex with positive incident edges, we can first flip the sign of the positive edges and when there are no more positive edges incident to $v$, remove the vertex $v$. Therefore, we can state the following theorem.
	\begin{thm}
		\label{thm:result:vertex:add:del:complexity}
		The procedure $\textsc{AddVertex}(v)$ ($\textsc{DeleteVertex}(v)$) stated here, adds a new vertex $v$ (deletes an existing vertex $v$) with all negative signed incident edges and maintains the graph and the clustering in $\bigoh{1}$ time using $\bigoh{1}$ working memory in addition to the memory required for representing clustering output, e.g. $\mathcal{C}_t$.
	\end{thm}
	Note that if the \textsc{Baseline} algorithm runs over $G_t$, the output is the same as $\mathcal{C}_t$ plus (minus) a singleton cluster containing $v$ if we add (delete) the vertex $v$ with all negatively signed incident edges.

	\section{Conclusion and Future Research Directions}
		In this paper, we proposed the first, up to the author's knowledge, online correlation clustering algorithm for dynamic complete signed graphs with full set of operations: (1) adding/deleting vertices, and (2) flipping the edge sign. The proposed approach employs the approximation scheme of \cite{clmnpt21-correlation-clustering-in-constant-many-parallel-rounds} and make it dynamic by studying its local behavior. The time and space complexity of the proposed approach is rigorously analyzed. In comparison with the \textsc{Baseline} algorithm, our proposed approach reduced the dependency of the running time in \textsc{Baseline} to the summation of the degree of all vertices in $G_t$ to the summation of the degree of the changing vertices (e.g. two endpoints of an edge) and the number of clusters in the previous time step. Moreover, using the locality effect, the required working memory of the proposed method is reduced to $\bigoh{\left(\deg_{G^+_t}(u) + \deg_{G^+_t}(v)\right)^2}$ compared with $\bigoh{\card{V_t}}$. A future research direction would applying the same technique to weighted graphs. Also, improving the memory and time requirements for this problem is another research direction. 

	\appendix
		\section{Appendix}
		\begin{lem}
	\label{lem:case:n:to:p:delta}
	Suppose the sign of edge $u=\set{u,v}$ is changed from $-$ to $+$ by a single call to $\textsc{FlipSign}_G(e)$. Then, 
	\begin{enumerate}
		\item $w \in N_{G^+}(u)$ implies $w \in N_{H^+}(u)$. %
		\item If $w \in N_{G^+}(u) \setminus N_{G^+}(v)$ and $w \neq v$, then $\card{N_{H^+}(u) \Delta N_{H^+}(w)} = \card{N_{G^+}(u) \Delta N_{G^+}(w)} + 1$.
		\item If $w \in N_{G^+}(u) \cap N_{G^+}(v)$, then $\card{N_{H^+}(u) \Delta N_{H^+}(w)} = \card{N_{G^+}(u) \Delta N_{G^+}(w)} - 1$.
		\item We have $\card{N_{H^+}(u) \Delta N_{H^+}(v)} = \card{N_{G^+}(u) \Delta N_{G^+}(v)} + 2$.
		\item Let $w \notin N_{G^+}(u) \cup N_{G^+}(v)$, then $\card{N_{H^+}(x) \Delta N_{H^+}(w)} = \card{N_{G^+}(x) \Delta N_{G^+}(w)}$ for all $x \in N_{G^+}(w)$.
	\end{enumerate}
	The same holds when $u$ and $v$ are exchanged in statements (1) to (3).
\end{lem}
\begin{proof}
	(1) The first claim is clear, since just the sign of edge $e=\set{u,v}$ is changed. 
	\\
	(2) By the hypothesis, we have $w \in N_{G^+}(u)$, however $w \notin N_{G^+}(v)$ (Figure \ref{fig.signflip}(a)). So, $v \in N_{H^+}(u) \Delta N_{H^+}(w)$ by the definition of the neighborhood operator. However, $v \notin N_{G^+}(u) \Delta N_{G^+}(w)$ (Figure \ref{fig.signflip}(b)). On the other hand, the only change made to the neighborhoods of vertices by this $\textsc{FlipSign}$ is for vertices $u$ and $v$. Therefore, the desired result is obtained.
	\\
	(3) In this case, the hypothesis implies $v \in N_{G^+}(w)$, since the graph is undirected. However, $v \notin N_{G^+}(u)$, and so, $v \in {N_{G^+}(u) \Delta N_{G^+}(w)}$ (Figure \ref{fig.signflip}(c)). By this call to $\textsc{FlipSign}$, we have $v \in N_{H^+}(u)$ which implies $v \notin {N_{H^+}(u) \Delta N_{H^+}(w)}$ (Figure \ref{fig.signflip}(d)). Similar to the second case, the neighborhood of no other vertex is updated, so the desired result is obtained. 
	\\
	(4) Note that neither $u \in N_{G^+}(v)$ nor $v \in N_{G^+}(u)$, since the sign of the edge $e=\set{u,v}$ is $-$ prior calling $\textsc{FlipSign}(e)$. However, after calling $\textsc{FlipSign}(e)$, we have $u \in N_{H^+}(v)$ nor $v \in N_{H^+}(u)$, therefore $u,v \in N_{H^+}(u) \Delta N_{H^+}(v)$. As nothing else changes in their neighborhood in $H^+$ compared to $G^+$, the claim holds.
	\\
	(5) As flipping the edge $e=\set{u,v}$ from $-$ to $+$ just changes the neighborhood of vertices $v$ and $u$, this claim holds. Note that $N_{G^+}(w) = N_{H^+}(w)$.
\end{proof}
Proof of Lemma \ref{lem:case:n:to:p:eps:agreement}:
\begin{proof}
	(1) By the third part of the Lemma \ref{lem:case:n:to:p:delta} and the definition of $\nonagreement{X}{u}{w}$ for $X=G^+,H^+$, we have $\card{N_{H^+}(u) \Delta N_{H^+}(w)} = \card{N_{G^+}(u) \Delta N_{G^+}(w)} + 1$ for all $w \in N_{G^+}(u) \cap N_{G^+}(v)$, i.e. the numerator of the $\nonagreement{H^+}{u}{w}$ is decreased by one compared to the numerator for $\nonagreement{G^+}{u}{w}$. Now, there are two cases to consider: (a) $N_{G^+}(u) < N_{G^+}(w)$, and (b) $N_{G^+}(u) \geq N_{G^+}(w)$. In the first case, (a), the denominator does not change in $\nonagreement{H^+}{u}{w}$ compared to $G^+$, so the claim holds. However, on the second case, (b), the denominator would increase by one, since $v$ is the new neighbor of $u$. This way, the value of $\nonagreement{H^+}{u}{w}$ is smaller than the one in $G^+$. Therefore, the desired result is obtained.
	\\
	(2) This case is very similar to the first case, since the numerator increases by one according to the second part of the Lemma \ref{lem:case:n:to:p:delta}, however, the denominator does not change. So, the desired result is obtained.
	\\
	(3) This case complements two previous cases for neighbors of one of the endpoints of the flipped edge, i.e. $w \in N_{G^+}(u) \setminus N_{G^+}(v)$, $\card{N_{G^+}(u)} \geq \card{N_{G^+}(w)}$, and $w \neq v$. In this case, both the numerator and the denominator increase by one. Hence, the value of $\nonagreement{G^+}{u}{w}$ can be bigger, smaller or even equal to $\nonagreement{H^+}{u}{w}$, depending the relation of $\textsc{Threshold}$ and $\card{N_{G^+}(u)}$. 
	\\
	(4) The numerator for $\nonagreement{H^+}{u}{w}$ is increased by two compared to $G^+$ (Lemma \ref{lem:case:n:to:p:delta}, case (4)) and the denominator increases by one, since $v$ is added as the new neighbors of vertex $u$ in $H^+$ and vice versa. Hence, based on the relation between $\textsc{Threshold}=\card{N_{G^+}(u) \Delta N_{G^+}(v)}$ and $2\max\set{\card{N_{G^+}(u)},\card{N_{G^+}(v)}}$, the claim is obtained.
	\\
	(5) By the fifth case of the Lemma \ref{lem:case:n:to:p:delta}, the numerator does not change in $\nonagreement{H^+}{u}{w}$ compared to $G^+$, however, the denominator might change, whether $N_{G^+}(w)$ is the maximum or $N_{G^+}(u)$. If $N_{G^+}(w)$ is the maximum, the denominator does not change, so the claim (5-a) is proved. If the maximum is $N_{G^+}(u)$, then the denominator is increased by one in $H^+$, so the claim (5-b) is obtained, too.
	\\
	(6) This is the direct consequence of case (6) in Lemma \ref{lem:case:n:to:p:delta} and the fact that no neighborhood, except for $u$ and $v$, are updated by flipping the sign of the edge $e=\set{u,v}$ from $-$ to $+$.
\end{proof}
Proof of Proposition \ref{prop:eps:agreement:status:n:to:p:neigh}
		\begin{proof}
	(1) Given that $u$ and $w$ are in $\varepsilon$-agreement in $G^+$, we know that $\nonagreement{G^+}{u}{w} < \varepsilon$. This pair of vertices remain $\varepsilon$-agreement in $H^+$ if $\nonagreement{H^+}{u}{w} < \varepsilon$, too. In the case (1-a), $\nonagreement{H^+}{u}{w} \leq \nonagreement{G^+}{u}{w} < \varepsilon$ by the first case of the Lemma \ref{lem:case:n:to:p:eps:agreement}. The last case, i.e. (1-b), hold due to case (3-c) of the Lemma \ref{lem:case:n:to:p:eps:agreement}.
	\\
	(2) Since $u$ and $w$ are not in $\varepsilon$-agreement in $G^+$, we have $\nonagreement{G^+}{u}{w} \geq \varepsilon$. The case (2-a) holds by the case (2) in Lemma \ref{lem:case:n:to:p:eps:agreement}, i.e. $\nonagreement{H^+}{u}{w} \geq \nonagreement{G^+}{u}{w} \geq \varepsilon$. The case (2-b) holds by the case (3-a) of the Lemma \ref{lem:case:n:to:p:eps:agreement}.
\end{proof}
Proof of Proposition \ref{prop:eps:agreement:status:n:to:p:rest}
		\begin{proof}
	This is a direct consequence of the definition of $\varepsilon$-agreement and the cases (4-a) to (4-c) of Lemma \ref{lem:case:n:to:p:eps:agreement}. Note that by case (4-b) of Lemma \ref{lem:case:n:to:p:eps:agreement}, $u$ and $v$ might be in $\varepsilon$-agreement in $G^+$, but not in $\varepsilon$-agreement in $H^+$. 
\end{proof}
		\begin{lem}
	\label{lem:case:p:to:n:delta}
	Suppose the sign of edge $u=\set{u,v}$ is changed from $+$ to $-$ by a single call to $\textsc{FlipSign}_G(e)$. Then, 
	\begin{enumerate}
		\item $w \in N_{G^+}(u)$ and $w \neq v$ implies $w \in N_{H^+}(u)$. %
		\item If $w \in N_{G^+}(u) \setminus N_{G^+}(v)$ and $w \neq v$, then $\card{N_{G^+}(u) \Delta N_{G^+}(w)} = \card{N_{H^+}(u) \Delta N_{H^+}(w)} + 1$.
		\item If $w \in N_{G^+}(u) \cap N_{G^+}(v)$, then $\card{N_{G^+}(u) \Delta N_{G^+}(w)} = \card{N_{H^+}(u) \Delta N_{H^+}(w)} - 1$.
		\item Let $w \notin N_{G^+}(u) \cup N_{G^+}(v)$, then $\card{N_{H^+}(x) \Delta N_{H^+}(w)} = \card{N_{G^+}(x) \Delta N_{G^+}(w)}$ for all $x \in N_{G^+}(w)$.
	\end{enumerate}
	The same holds when $u$ and $v$ are exchanged in statements (1) to (3).
\end{lem}
\begin{proof}
	The proof is similar to the proof of Lemma \ref{lem:case:n:to:p:delta}, i.e. swapping $G^+$ and $H^+$ in the proofs. 
\end{proof}
Proof of Lemma \ref{lem:case:p:to:n:eps:agreement} 
		\begin{proof}
	(1) By the third case of Lemma \ref{lem:case:p:to:n:delta}, we have the numerator of $\nonagreement{H^+}{u}{w}$ increases by one compared to $\nonagreement{G^+}{u}{w}$. The denominator does not change in $H^+$ if $\card{N_{G^+}(w) \geq \card{N_{G^+}(u)}}$, otherwise it would also decrease by one. So, the $\nonagreement{H^+}{u}{w} > \nonagreement{G^+}{u}{w}$. 
	\\
	(2,3) The proof of these parts is very similar to the second and the third parts in Lemma \ref{lem:case:n:to:p:eps:agreement}. Note the slight difference in the hypothesis of these parts compared to the case 2 and 3 in Lemma \ref{lem:case:n:to:p:eps:agreement}, respectively. 
	\\
	(4) This is the direct implication of the last part of Lemma \ref{lem:case:p:to:n:delta} and the fact that no neighborhood other than $u$ and $v$ have changes in $H^+$. 
\end{proof}
Proof of Proposition \ref{prop:eps:agreement:status:p:to:n:neigh}
		\begin{proof}
	The proof is similar to the proof of the Proposition \ref{prop:eps:agreement:status:n:to:p:neigh}, using Lemma \ref{lem:case:p:to:n:eps:agreement}. 
\end{proof}
Proof of Proposition \ref{prop:eps:agreement:status:p:to:n:rest} 
		\begin{proof}
	This can be shown by a similar argument to the Proposition \ref{prop:eps:agreement:status:n:to:p:rest} based on the Lemma \ref{lem:case:p:to:n:eps:agreement}.
\end{proof}
Proof of Corollary \ref{cor:agreecount:not:changing:non:neighbors}
		\begin{proof}
	Let the sign of edge $e$ changed from $-$ to $+$. Assume that $w \notin S$ and $x$ be an arbitrary neighbor of $w$ in $G^+$, i.e. $x \in N_{G^+}(x)$. Given that $x$ and $w$ are in $\varepsilon$-agreement in $G^+$, then by the Proposition \ref{prop:eps:agreement:status:n:to:p:rest}, case (4), they would be in $\varepsilon$-agreement in $H^+$, too. If they are not in $\varepsilon$-agreement in $G^+$, nothing changes in their $\varepsilon$-agreement in $H^+$. Moreover, the degree of $w$ in $G^+$ and $H^+$ is the same. Therefore, the claim is proved for flipping from $-$ to $+$. A similar argument using the Proposition \ref{prop:eps:agreement:status:p:to:n:rest} shows the claim holds for flippinf the edge sign from $+$ to $-$. Hence, the claim is proven. 
\end{proof}
Proof of Theorem \ref{thm:lightness:edge:flip}
\begin{proof}
	(1,2) Given that $w \notin S$, the fourth case of Proposition \ref{prop:eps:agreement:status:n:to:p:rest} implies that the number of vertices which are in \epsagree with $w$ in $H^+$ does not change for this vertex in comparison with $G^+$. Therefore, if $w$ is not a $\varepsilon$-light vertex in $G^+$, it remains non-$\varepsilon$-light in $H^+$, too. If $w$ is a $\varepsilon$-light vertex in $G^+$, then by the same argument it remains a $\varepsilon$-light vertex in $H^+$, also.
	\\
	(3) This is simply implied by Table \ref{tbl:edge:agreement:summary} and Corollary \ref{cor:agreecount:not:changing:non:neighbors}.
\end{proof}
Proof of Theorem \ref{thm:flip:edge:time:complexity:updating:g:plus}
		\begin{proof}
	Considering the discussion just before this theorem for time and space analysis, one can simply call either Algorithm \ref{alg:update:g:plus:flip:sign:n:to:p} or \ref{alg:update:g:plus:flip:sign:p:to:n}, whether the sign of edge $f$ is changed from $-$ to $+$ or $+$ to $-$, respectively. 
\end{proof}

		\begin{algorithm}[ht]
	\caption{$\textsc{Update}G^+\textsc{PositiveToNegative}$ procedure runs after every call to $\textsc{FlipSign}_G(e)$ which changed $e.\textsc{Sign}$ from $+$ to $-$.}
	\label{alg:update:g:plus:flip:sign:p:to:n}
	\begin{algorithmic}[1]
		\Procedure{Update$G^+\textsc{PositiveToNegative}$}{$e=\set{u,v}$}
		\State Let $G_t$ be the graph after the $t^{th}$ modification command which was $\textsc{FlipSign}_G(e)$. 
		\State $G^+ \gets G_t[E^+]$
		\State Let $f=\set{u,v}$
		\If{$f.\textsc{Agree}$} \Comment{After deleting $f$, \epsagree no more means for $u$ and $v$}
		\State $u.\agreecnt_{G^+} \gets u.\agreecnt_{G^+} - 1$
		\State $v.\agreecnt_{G^+} \gets v.\agreecnt_{G^+} - 1$
		\EndIf
		\State Let $A \gets N_{G^+}(u) \cap N_{G^+}(v)$
		\State Let $B=\left(N_{G^+}(u) \setminus N_{G^+}(v)\right) \setminus \set{v}$
		\State Let $C=\left(N_{G^+}(v) \setminus N_{G^+}(u)\right) \setminus \set{u}$
		\ForAll{$w \in A$} \Comment{Applying Lemma \ref{lem:case:p:to:n:eps:agreement}, (1)}
		\State Let $f=\set{u,w}$ and $g=\set{v,w}$
		\If{$f.\textsc{Agree}$}
		\State Call \textsc{VerifyEdge}($f=\set{u,w}$)
		\EndIf
		\If{$g.\textsc{Agree}$}
		\State Call \textsc{VerifyEdge}($g=\set{v,w}$)
		\EndIf
		\EndFor
		\For{$(x, X)$ in $[(u, B), (v, C)]$} 
		\State \Comment{Run once with $x\gets u$ and $X\gets B$, another time with $x\gets v$ and $X\gets C$}
		\ForAll{$w \in X$}
		\State Let $f=\set{x,w}$
		\If{$f.\textsc{Agree}$}
		\If{$\deg_{G^+}(x) \leq \deg_{G^+}(w)$} \Comment{Applying Lemma \ref{lem:case:p:to:n:eps:agreement}, (2)}
		\State Call \textsc{VerifyEdge}($f=\set{x,w}$)
		\ElsIf{$\card{N_{G^+}(x) \Delta N_{G^+}(w)} > \deg_{G^+}(x)$} \Comment{Applying Lemma \ref{lem:case:p:to:n:eps:agreement}, (3-c)}
		\State Call \textsc{VerifyEdge}($f=\set{x,w}$)
		\EndIf
		\Else \Comment{$f.\textsc{Agree} = \textsc{False}$}
		\If{$\deg_{G^+}(x) > \deg_{G^+}(w)$ \textbf{ and } $\card{N_{G^+}(x) \Delta N_{G^+}(w)} < \deg_{G^+}(x)$} 
		\State \Comment{Applying Lemma \ref{lem:case:p:to:n:eps:agreement}, (3-a)}
		\State Call \textsc{VerifyEdge}($f=\set{x,w}$)
		\EndIf
		\EndIf
		\EndFor
		\EndFor
		\State Let $S \gets \set{u,v} \cup u.\textsc{Neigh}_{G^+} \cup v.\textsc{Neigh}_{G^+}$
		\ForAll{$w \in S$} \Comment{Fixing the lightness of the vertices in $G^+$}
		\State $w.\textsc{IsLight} \gets \left(\frac{w.\agreecnt_{G^+}}{\deg_{G^+}(w)} < \varepsilon\right)$
		\EndFor
		\State \Return $G^+$ as $H^+$.
		\EndProcedure
	\end{algorithmic}
\end{algorithm}
Proof of Lemma \ref{lem:flip:sign:change:sparsified:graph}
	\begin{proof}
	The discussion just before the Lemma \ref{lem:flip:sign:change:sparsified:graph} is its proof.
\end{proof}
Proof of Lemma \ref{lem:flip:sign:prune:state:space}
	\begin{proof}
	This is a direct consequence of Lemma \ref{lem:flip:sign:change:sparsified:graph}, since all the edges inside $\widetilde{G_{t-1}^+}[C]$ does not belong to $S$. Therefore, $\widetilde{G_{t-1}^+}[C]=\widetilde{G_{t}^+}[C]$ and the claim is shown. 
\end{proof}
Proof of Theorem \ref{thm:cluster:maintain:flip:sign}
	\begin{proof}
	For flipping the sign of an edge $e=\set{u,v}$, we first compute the connected components of $\widetilde{G_{t-1}^+}[S]$ where $S=\set{u,v} \cup N_{G_{t-1}^+}(u) \cup N_{G_{t-1}^+}(v)$ and then, decide to merge, split or simply transfer each cluster in $\mathcal{C}_{t-1}$. An split possibility is taken into account whenever there are two vertices $w,w' \in S$ which are in different connected components. As they are now separated in $S$, we need to verify their connection by paths out of $S$. Therefore, we compute the connected components of $\widetilde{G_{t-1}^+}[C]$ where all edges with both endpoints in $S$ removed. If there is one connected component, then no split occurs as the vertices $w$ and $w'$ are still connected by paths with no intersection with $S$. However, we need to fall-back to our merging decision if there are any or even copy decision. On the other hand, if there are more than one connected components, we split $C$ into its connected components. Next, we need to take care of merging these split new clusters with existing ones. To do this, we fall-back again to cluster groups and assign the new split clusters either a copy decision or a merge decision. If take a high-level look, in fact we are reconstructing the connected components of the graph $\widetilde{G_{t}^+}$, which is exactly the same as running the \textsc{Baseline}.
	\\
	For the time complexity, both steps one and two in marking can be accomplished in time $\card{\mathcal{C}_{t-1}}$. The third step can be accomplished in $\bigoh{\card{S}\times\card{S}}$, too. Moreover, we need to compute the connected components of the clusters with split possibility as well as the subgraph with vertices in $S$, which accumulation is $\bigoh{\card{V_t}+\card{E^+_t}}$. Therefore, the claimed running time is obtained. 
	\\ 
	For the space complexity, consider storing cluster groups which are $\bigoh{\card{S}\times\card{S}}$ in the worst case. Each pair is used for at most two clusters, so the claimed space complexity is obtained. 
\end{proof}
		
	\bibliographystyle{alpha}%
	\bibliography{references}
\end{document}